\documentclass[12pt,twoside]{amsart} 
\usepackage[papersize={210mm,297mm},left=15mm,right=15mm,top=20mm,bottom=20mm]{geometry}
\usepackage{graphicx}
\usepackage{mathrsfs}
\usepackage{color}
\usepackage{amssymb}
\usepackage{amsmath}
\usepackage{amsthm}
\usepackage{amssymb}
\usepackage{xcolor}
\usepackage{todonotes}
\definecolor{slategray}{RGB}{112,138,144}
\usepackage{tikz,pgfplots}
\pgfplotsset{compat=1.7}
\newif\ifintrmk
\intrmktrue
\newcommand{\intremark}[1]{\par\bigskip {\color{navy}\noin{\bf Internal Remark:} #1} \par\bigskip}
\usepackage{xcolor}
\definecolor{slategray}{RGB}{112,138,144}
\definecolor{hintergrundblau}{RGB}{0 ,43 ,54}
\definecolor{wichtigrot}{RGB}{220,50,47}
\definecolor{textgrau}{RGB}{131,148,150}
\definecolor{ultramarinblau}{RGB}{32,33,79}
\definecolor{kobaltblau}{RGB}{35,45,83}
\definecolor{navy}{RGB}{0 ,0,128}
\definecolor{leuchtgruen}{RGB}{0,181,26}
\definecolor{leuchtrot}{RGB}{255,77,6} 
\definecolor{structuresgray}{RGB}{162,173,191}

\usepackage[colorlinks=true]{hyperref}
\hypersetup{linkcolor=slategray,urlcolor=slategray,citecolor=slategray}
\usepackage[preserveurlmacro]{breakurl}

\newcommand{\E}{{\rm e}}
\newcommand{\I}{{\rm i}}
\newcommand{\dd}{{\rm d}}
\newcommand{\beq}{\begin{equation}}
\newcommand{\eeq}{\end{equation}}
\newcommand{\neeq}{\nonumber\end{equation}}

\newcommand{\Tr}{{\rm Tr }}

\newcommand{\cH}{{\mathcal H}}

\newcommand{\ps}{\psi}

\newcommand{\psq}{{\bar\psi}}

\newcommand{\veps}{\varepsilon}

\usepackage{graphicx}
\usepackage{appendix}
\usepackage{cases}

\newtheorem{theorem}{\textsc{Theorem}}[section]
\newtheorem{lemma}{\textsc{Lemma}}[section]
\newtheorem{definition}{\textsc{Definition}}[section]

\newtheorem{remark}{\textsc{Remark}}[section]

\numberwithin{equation}{section}

\newcommand{\reftitel}[1]{{\sl #1}}

\newcommand{\bracket}[2]{\left\langle #1 \mid #2 \right\rangle}
\newcommand{\abs}[1]{\left|#1\right|}
\newcommand{\norm}[1]{\left\Vert#1\right\Vert}

\DeclareMathSymbol{\shortminus}{\mathbin}{AMSa}{"39}

\newcommand{\rd}{{\textup{d}}}

\newcommand{\R}{{\mathbb{R}}}
\newcommand{\C}{{\mathbb{C}}}

\newcommand{\N}{{\mathbb{N}}}

\newcommand{\cF}{{\mathcal{F}}}
\newcommand{\W}{{\mathcal{W}}}
\newcommand{\X}{{\mathbb{X}}}

\newcommand{\cG}{{\mathcal{G}}}
\newcommand{\cE}{{\mathcal{E}}}

\newcommand{\V}{{\mathcal{V}}}

\newcommand{\del}[1]{\frac{\partial}{\partial #1}}
\newcommand{\delt}[1]{\frac{\delta}{\delta #1}}

\newcommand{\ex}[1]{ \left\langle #1 \right\rangle }
\newcommand{\gint}[1]{\int \rd \mu_{#1}(\bar{\psi},\psi)}

\newcommand{\mynorm}[1]{{\left\vert\kern-0.25ex\left\vert\kern-0.25ex\left\vert #1 
    \right\vert\kern-0.25ex\right\vert\kern-0.25ex\right\vert}}
    \newcommand{\eunorm}[1]{ \left\| #1 \right\| }

\newcommand{\Xpos}{{\bf X}} 

\newcommand{\gammaT}{\gamma^{\rm c}} 
\newcommand{\gammaTz}{{}_{_0}\mkern-5mu\gamma^{\rm c}}

\newcommand{\Time}{T} 

\newcommand{\noin}{\noindent}

\renewcommand{\intremark}[1]{}
\begin{document}

\title[A rigorous Keldysh functional integral for fermions]{A rigorous Keldysh functional integral for fermions}
\author{Philipp Benjamin Aretz, Manfred Salmhofer} 
\address{%
Institut f\" ur Theoretische Physik, Universit\" at Heidelberg,
Philosophenweg~19, 69120 Heidelberg, Germany}

\email{aretz24@mit.edu,salmhofer@uni-heidelberg.de}

\date{\today. Based on the first author's bachelor thesis in physics at Heidelberg University}

\begin{abstract}
\noindent
We provide a mathematically rigorous Keldysh functional integral for fermionic quantum field theories. We show convergence of a discrete-time Grassmann Gaussian integral representation in the time-continuum limit under very general hypotheses. We also prove analyticity of the effective action and explicit bounds for the truncated (connected) expectation values $\gammaT_{m,\bar{m}}$ of the non-equilibrium system. These bounds imply clustering with a summable decay in the thermodynamic limit, provided these properties hold at time zero, and provided that the determinant bound $\delta_C$ and decay constant $\alpha_C$ of the fermionic Keldysh covariance are bounded uniformly in the volume. We then give bounds for these constants and show that uniformity in the volume indeed holds for a general class of systems. Finally we show that in the setting of dissipative quantum systems, these bounds are not necessarily restricted to short times.
\end{abstract}

\maketitle

\section{Introduction}
In his 1964 paper \cite{Keldysh}, Leonid Keldysh introduced a formalism for non-equilibrium quantum field theory that now carries his name (it is often also called Schwinger-Keldysh formalism, motivated by work of Julian Schwinger \cite{Schwinger},  see e.g.\ \cite{Kamenev} for a review). The analysis of functional integrals on the closed-time contour has since become one of the standard approaches to nonequilibrium phenomena in theoretical physics. In this paper, we construct the Keldysh functional integral for fermions in a simple but rigorous way, generalizing the treatment of equilibrium expectation values of \cite{Salmhofer:2009wm}. The construction is general, applying to hermitian as well as non-hermitian Hamiltonians, and to general even fermionic interactions $V$. The ultraviolet problem arising from the discontinuity of the time-ordered covariance is treated using the results of \cite{Pedra2008,Salmhofer2000}, see also \cite{dRS}. Under a condition on the spatial decay of the correlations for non-interacting fermions ($V=0$), we prove decay bounds for the truncated correlation functions uniformly in the system size. 

Our motivation to make the Keldysh functional integral rigorous is to make the full formalism of quantum field theory (QFT) available for rigorous studies of interaction effects in a time-dependent setting. Specifically, the present paper provides a setup for applying functional integral methods of constructive QFT to transport problems, such as the question of the emergence of the quantum Boltzmann equation (QBE) in the kinetic limit \cite{ESY,LukkSpohn,Lukkarinen}. A nice feature of the argument in \cite{ESY}  is that the form of the collision term for quantum systems comes out easily, without any resummation of graphs (as done, e.g.\ in \cite{Hugenholtz}), and that thereby, the proof of convergence to the QBE in the kinetic limit is translated to showing that the truncated $n=4$ and $n=6$-point correlation functions of a many-fermion system, are small for weak interaction $V$ and vanish in the limit $V\to 0$, in a certain norm. The setup given here allows to pursue this strategy: under the just stated conditions, the Schwinger-Dyson equation on the Keldysh contour implies the QBE, in a manner similar to the result of \cite{ESY}, and the conditions on the four- and six-point functions can be studied by fermionic functional integral techniques. 
This would, however, be beyond the scope of this paper and is left for future work. 

The Keldysh formalism is also used for many-boson systems, and one may ask why we concentrate on fermions here.  Fermionic theories have technical advantages over bosonic ones: fermion operators are bounded (the fermionic ladder operators satisfy $a(f) \le \norm{f}$ where $\norm{\cdot}$ is the norm on the one-particle Hilbert space $\cH$). In the Grassmann representation, they are given by nilpotent variables, and there is no convergence issue for Gaussian fermionic expectations, unlike bosonic ones. That is, if $A$ is invertible, then $C=A^{-1}$ is given by 
\beq
C_{x,y} 
=
(\det A)^{-1} \; 
\int \dd\psq\dd\ps\; \E^{-(\psq,A\ps)} \; \ps_x \psq_y
=
\int \dd\mu_C (\psq,\ps) \; \ps_x \psq_y
\eeq
irrespective of any positivity properties of the operator $A$ or its hermitian part (which would be required for the convergence of the integral if the $\ps$ were complex variables). Thus fermionic Gaussian integrals can be defined for more general covariance operators than bosonic ones. This is advantageous when one wants to study unitary time evolutions at real time, not just equilibrium expectation values. (Oscillatory functional integrals can be defined mathematically and have been studied, see, e.g.\ \cite{Mazzucchi}.)

Positivity conditions do play an important role when cumulants (connected correlation functions) are considered, but they concern not the covariance itself, but more technical objects used in proofs: in order to obtain good determinant bounds, one needs to preserve positivity of weighting matrices arising from interpolation. This is explained in detail in \cite{Salmhofer2000}. 

From the point of view of analysis, an essential purpose of this paper is to make explicit that the determinant bounds of \cite{Pedra2008} and their optimized form of \cite{dRS} also apply at real time. An important challenge when going from imaginary to real times is to give good bounds for the decay constants that also apply to infinite systems. In this context we discuss systems with hermitian and non-hermitian one-particle Hamiltonians. The latter are frequently used as an effective description of dissipative systems (the no-jump limit of Lindblad dynamics) \cite{Rotter}. If the non-hermitian part of the Hamiltonian is strictly negative (as suggested by standard heuristic second-order perturbative arguments), the decay constant has a much better behaviour as a function of time than in the hermitian case, see Lemma \ref{lem:determinant_bound_and_decay_const}. Ultimately, and especially for understanding the Boltzmann equation, we are interested in the hermitian case, however, and the non-hermitian part needs to be sent to zero. This is a particular choice of an $\I \veps$ condition.

\section{General Setup and Decay of Correlations}\label{sec:setupresult}
Let $ \varepsilon > 0$ and $\cH$ be a finite-dimensional Hilbert space with orthogonal basis $\{x\}_{x\in \Xpos}$, normalized to $\bracket{x}{x}= \varepsilon^{-1}$, associated fermionic Fock Space $\cF = \bigwedge\cH$ and (creation) annihilation operators $a^{(*)}$ defined in the standard way \cite{BratRob}. In applications, $\Xpos$ often corresponds to a discrete $d$-dimensional lattice with appropriate lattice spacing (spacing of $\varepsilon^{\frac{1}{d}}$ for a cubic lattice), finite side length $L$ and additional spin indices. Alternatively, $\Xpos$ simply labels the orthonormal basis of a quantum system with Hilbert space $\cH$ (in which case $\varepsilon =1$). For the sake of generality, we will use continuum notation, i.e. $(a^*, a)_{\Xpos} =\int_{\Xpos}\rd x \ a^*(x) a(x) = \varepsilon \sum_{x \in \Xpos} a^*(x) a(x)$ and $\delt{\bar{\psi}(x)} = \varepsilon^{-1} \del{\bar{\psi}(x)}$, where $\bar{\psi}$ is part of the Grassmann algebra isomorphic to $\cF$ generated by $\{\bar{\psi}(x)\}_{x \in \Xpos}$. We write $\C\{\bar{\psi}(x)\}_{x\in \Xpos}$ for this algebra. More information on the notation used and an introduction to calculus on Grassmann algebras and especially Grassmann Gaussian integration can be found in \cite[Appendix~B]{MSbook}. For a more abstract introduction to Grassmann algebras and their calculus see \cite[Section~1]{Feldman:2002}.

Our main goal is to generalize the proof of a convergent functional integral representation for equilibrium correlation functions of many-fermion systems presented in \cite{Salmhofer:2009wm} to time-dependent expectation values of observables $O \in \mathcal{L}(\mathcal{F},\mathcal{F})$, with the time evolution governed by a not necessarily self-adjoint Hamiltonian $H$ and the initial state given by a density operator $\frac{1}{Z_0}\rho_0$ with $\rho_0=\E^{-\beta(a^*,Qa)_{\Xpos}}$, with hermitian matrix $Q$, $\beta > 0$ and normalisation $Z_0 = \Tr(\rho_0)$. For time $\Time >0$, 
\begin{align}\label{eq:ex_A}
	\ex{O}_{\Time} = \frac{1}{Z_0} \Tr \left( \E^{iH^{\dag}\Time} O \E^{-iH\Time} \rho_0 \right).
\end{align}
For our finite fermionic system, $H$ can always be written as a polynomial in $a$ and $a^*$. We assume that $H=H_0+V$ splits into a quadratic part $H_0$ and a higher degree part $V$, with the usual interpretation that the special case $V=0$ models independent fermions and $V\ne 0$ means that they interact. We further assume that $H_0$ is of the form 
\beq\label{H0AB}
H_0=(a^*,(A-iB)a)_{\Xpos}
\eeq 
with $B\geq0$ and $A,B$ hermitian. Additionally, we assume $V$ to be normal. The potential $V$ is uniquely described by its interaction vertices $v_{m,\bar{m}}$, defined as the coefficients of the normal ordered form of $V$:
\begin{align}\label{eq:def_V}
	\begin{aligned}
	V&=\sum_{m,\bar{m} \in \N} \int_{\Xpos} \prod_{i=1}^{m} \rd x_i  \int_{\Xpos} \prod_{j=1}^{\bar{m}} \rd y_j \  v_{m,\bar{m}}(x_1,\dots,x_m;y_1,\dots,y_{\bar{m}}) \prod_{j=1}^{\bar{m}} a^*(y_j)  \prod_{i=1}^{m}a(x_i)\\
	&= \sum_{m,\bar{m} \in \N} \int_{\Xpos} \rd^m x  \int_{\Xpos} \rd^{\bar{m}} y \ v_{m,\bar{m}}({x};{y}) \left(a^*\right)^{\bar{m}}({y}) a^m({x}).
	\end{aligned}
\end{align}
From now on, for the sake of conciseness, we will always use short-hand notation similar to that in the second line of \eqref{eq:def_V}. 
A typical case to model would be a perturbed grand-canonical ensemble. This is achieved by setting $Q=A$ and making $H$ self-adjoint by setting $B=0$ and $V$ to be hermitian. Then $V$ acts as perturbation introduced at time $t=0$. 
Instead of focussing on the expression for $\ex{O}_{\Time}$ directly, we will establish a functional integral for the generating functional
\begin{align}\label{eq:gen_fct}
	Z(c^-,c^+,\Time)  = \Tr(\E^{iH^{\dag}\Time}\E^{(c^+,a^*)_{\Xpos}}\E^{(c^-,a)_{\Xpos}}\E^{-iH\Time}\rho_0)
\end{align}
of the reduced density matrices
\begin{align}\label{eq:red_dens}
	\begin{aligned}
	\gamma_{m,\bar{m}}(x_1, \dots,x_m,y_1, \dots y_{\bar{m}},\Time)  &=\ex{\prod_{i=1}^ma ^{*}(x_i)\prod_{j=\bar{m}}^1a(y_j) }_{\Time}\\
	&= \frac{1}{Z_0} \Tr \left( \E^{iH^{\dag}\Time} \prod_{i=1}^ma ^{*}(x_i)\prod_{j=\bar{m}}^1a(y_j) \E^{-iH\Time} \rho_0 \right)\\
	&=\frac{\delta^m}{{\delta c_{x_1}^+\dots \delta c_{x_m}^+}} \frac{\delta^{\bar{m}}}{\delta c_{y_{\bar{m}}}^- \dots \delta c_{y_1}^-} \frac{Z(c^-,c^+)}{Z(0,0)} \Big|_{c^+=0,c^-=0},
	\end{aligned}
\end{align}
where $c^-(x),c^+(x) \in \C\{c^+(x),c^-(x)\}_{x \in \Xpos}$ are Grassmann-valued source fields\footnote{It is easy to construct $c^\pm$ by using the isomorphy of $\cF \wedge \cF \wedge \cF$ 
to a larger Grassmann algebra in which the $c^\pm$ appear as generators that correspond to the first two factors in the exterior product} 
that anticommute with $a$ and $a^*$: 
\beq
\{c^{\pm}(x)\;,\;a^{*}(y)\} = \{c^{\pm}(x)\;,\;a(y)\} = 0 \;.
\neeq
In \eqref{eq:red_dens}, the ``evaluation at $c^+=0$ and $c^-=0$'' means that the Grassmann element is projected to its ``body'', i.e.\ the constant term in the polynomial $f(c^+,c^-)$.
The expectation value of an arbitrary operator $O \in \mathcal{L}(\cF,\cF)$ can then be written as a linear combination of the reduced density matrices, and as such a convergent functional integral representation of $Z$, implies the same for $\ex{O}_{\Time}$. 

To establish quantitative bounds on the evolution of the underlying fermionic systems, we will inspect the connected parts of the reduced density matrices (also called truncated expectation values or cumulants)  $\gammaT_{m,\bar{m}}$, which are
\begin{align}\label{eq:tr_exp_values}
 \gammaT_{m,\bar{m}} (x_1,\dots,x_m,y_1,\dots ,y_{\bar{m}},\Time) = \frac{\delta^m}{{\delta c_{x_1}^+\dots \delta c_{x_m}^+}} \frac{\delta^{\bar{m}}}{\delta c_{y_{\bar{m}}}^- \dots \delta c_{y_1}^-} F(c^-,c^+)\Big|_{c^+,c^-=0},
\end{align}
with the cumulant generating functional 
\begin{align}\label{eq:def_F}
	F(c^-,c^+,\Time)  = \log(\frac{Z(c^-,c^+,\Time) }{Z_0})
\end{align}
which can be viewed formally as an analogue of a free energy. 

In the following we will always assume that $V$ is even, that is, $v_{m,\bar m} = 0$ if $m+\bar{m}$ is an odd integer. As the constant terms in the Hamiltonian $H$ cancel out in \eqref{eq:gen_fct}, we may also assume w.l.o.g. that $v_{0;0}=0$. The $v_{m,\bar m}$ are taken to be antisymmetric under permutation of the $x$ and the $y$ variables. 

Under the additional assumption of $U(1)$-invariance of $V$ under $a(x) \mapsto \E^{\I \alpha} a(x)$,  our entire system is $U(1)-invariant$, and then $m\ne\bar{m}$ implies $\gammaT_{m,\bar{m}} = 0$. We will, however, not impose this symmetry in this paper. We also do not need translation invariance. 

We will be interested in weakly interacting systems, where the data in $H_0$ determine essential properties. In Section \ref{sec:Keldysh}, we define covariance $C$, its determinant constant $\delta_C$, and its decay constants $\alpha_C$ and $\tilde\alpha_C$. All of these constants a priori depend on total time $\Time$, system volume $|\Xpos|$, the lattice spacing $\varepsilon^{\frac{1}{d}}$, and the one-particle operators $A,B$ and $Q$. We do, however, give bounds uniform in $\Xpos$ for these constants, and we prove the following theorem. 

The interaction defines a sequence of functions $(v_{m,\bar m})_{m,\bar m: m+\bar m > 0}$. For $h > 0$, we define a norm on the space of such sequences by 
\beq\label{eq:oneinf}
\eunorm{V}_{h}
=
\sum_{m,\bar m \ge 0 \atop m+\bar m > 0} \abs{v_{m,\bar m}}_{1,\infty,\Xpos} \; h^{m+\bar m}
\eeq
where $\abs{v_{m,\bar m}}_{1,\infty,\Xpos} = \max \{ \sup_{x\in \Xpos} \lambda (x) , \sup_{y\in \Xpos} \tilde \lambda (y)\}$, with
\beq
\lambda(x) = \int_{\Xpos^{m-1}} \rd^{m-1} x'  \int_{\Xpos} \rd^{\bar{m}} y \ \abs{v_{m,\bar{m}}({x,x'};{y})} \; ,
\qquad
\tilde \lambda(y) = \int_{\Xpos^{m}} \rd^{m} x'  \int_{\Xpos^{\bar m -1}} \rd^{\bar{m}-1} y' \ \abs{v_{m,\bar{m}}({x'};{y,y'})} 
\eeq
In case of a translation invariant system, $\lambda(x)$ does not depend on $x$ and $\lambda (y)$ does not depend on $y$ any more, so then  $\abs{v_{m,\bar m}}_{1,\infty,\Xpos} = |\Xpos|^{-1} \; \norm{v_{m,\bar m}}_{1}$ (with all summations in the $1$-norm taken over $\Xpos$).

We use the same norms for the sequence of truncated reduced density matrices $\gammaT_{m,\bar m}$. We also denote the $\gammaT $ for $V=0$ by $\gammaTz$.

\begin{theorem}\label{th:main_theorem}
Let $V$ be an even interaction and $\eunorm{\cdot}_h$ be defined as in  \eqref{eq:oneinf}.
Let $C$ be the Keldysh covariance associated to $H_0$ and $\rho_0$, and $\alpha_C$ and $\tilde \alpha_C$ be its decay constants (see Definition \ref{def:decay_const}), and $\delta_C$ be its determinant constant (see Lemma \ref{lem:det_const}),  
and $\omega_C=2 \alpha_C \delta_C^{-2}$.  
If $\omega_C \eunorm{V}_{3\delta_C,{\Xpos}} \leq \frac{1}{2}$, then 
\begin{align}\label{abfi}
	 \abs{\gammaT_{m,\bar{m}}-\gammaTz_{m,\bar{m}}}_{1,\infty, {\Xpos}}  \leq 2 m! \bar{m}! \tilde{\alpha}_C^{m+\bar{m}-1} \alpha_C \delta_C^{-m-\bar{m}} \eunorm{V}_{3\delta_C,{\Xpos}}.
\end{align}
For $m\neq 1$ or $\bar{m}\neq1$, the free truncated expectation values $\gammaTz_{m,\bar{m}}$ vanish, and \eqref{abfi} becomes
\begin{align}
	\abs{\gammaT_{m,\bar{m}}}_{1,\infty,{\Xpos}}\leq 2 m! \bar{m}! \tilde{\alpha}_C^{m+\bar{m}-1} \alpha_C \delta^{-m-\bar{m}} \eunorm{V}_{3\delta_C,{\Xpos}} .
\end{align}
The same bounds hold for truncated expectation values of unordered monomials for all $m + \bar{m} \neq 2$.
\end{theorem}

Theorem \ref{th:main_theorem} is proven in Section 3, using the convergent Keldysh functional integral derived there. It implies that if the determinant and decay constants are uniform in $\Xpos$, the truncated reduced density matrices (cumulants) are absolutely summable, if one of the external variables is fixed (thus ``have $\ell^1$-decay'').

\section{The Keldysh Functional Integral} \label{sec:Keldysh}

The derivation of the functional integral of $Z$ is similar to the derivation undertaken in \cite[Appendix]{Salmhofer:2009wm}. We will only outline the main strategy here. All steps left out are straightforward although often a bit lengthy. 

\subsection{Construction of the Functional Integral}
We use the variant of the Lie product formula given in \cite[Lemma~3]{Salmhofer:2009wm} to discretize the time evolution in \eqref{eq:gen_fct}. Defining
\begin{align}
C_N = \E^{-i\frac{\Time}{N}H_0}, &\quad D_N = (1-i\frac{\Time}{N}V),
\end{align}
this implies $\E^{-iH\Time} = \lim_{n\rightarrow \infty}\left((C_N D_N)^N\right)$ and thus
\begin{align}\label{eq:first_def_Z_N}
	\begin{aligned}
	Z(c^-,c^+,\Time)   =& \lim_{N\rightarrow \infty} \Tr((C_N^{\dag} D_N^{\dag})^N \E^{(c^+,a^*)_{\Xpos}}\E^{(c^-,a)_{\Xpos}} (C_N D_N)^N\rho_0)\\
	 =&\lim_{N \rightarrow \infty} Z_N(c^-,c^+).
	\end{aligned}
\end{align}
$Z_N$ may then be understood as the time-discretized version of the generating function $Z$. In the process, we lost the exponential character of the time evolution in $V$. While this helps by making the replacement with Grassmann variables simpler, we will later on need to "re-exponentiate" all terms containing the potential $V$.

Utilising the isomorphism between $\cF$ and the Grassmann algebra $\{\bar{\psi}(x)\}_{x \in \Xpos}$, \cite[Lemma~B.18]{MSbook} allows us to express traces of operators $A\in \mathcal{L}(\cF,\cF)$ via Grassmann integrals. Explicitly,
\begin{align}
	\begin{aligned}
		\Tr(A) = \int D_{\Xpos}(\bar{\psi},\psi) \cG(A)(-\bar{\psi},\psi)\E^{-(\bar{\psi},\psi)_{\Xpos}} = \int D_{\Xpos}(\bar{\psi},\psi) \cG(A)(\bar{\psi},-\psi)\E^{-(\bar{\psi},\psi)_{\Xpos}}.
	\end{aligned}
\end{align}
Here $\int D_{\Xpos}(\bar{\psi},\psi) = \varepsilon^{-|\Xpos|} \int D\bar{\psi} D\psi = \varepsilon^{-|\Xpos|} \prod_{x \in \Xpos} \int d\bar{\psi}(x) d\psi(x)$ is used for notational simplicity. The Grassman symbol is defined as $\cG(A)=\E^{(\bar{\psi},\psi)_{\Xpos}}\Omega(A)(\bar{\psi},\psi)$, where $\Omega(A)(a^*,a)$ refers to the normal ordered form of $A$. The Grassmann symbol of an operator product is
\begin{align}
		\cG(AB)(\bar{\psi},\psi) = \int D_{\Xpos}(\bar{\eta},\eta) \cG(A)(\bar{\psi},\eta) \E^{-(\bar{\eta},\eta)_{\Xpos}} \cG(B)(\bar{\eta},\psi).
\end{align}
Making the appropriate replacements in \eqref{eq:first_def_Z_N}, writing 
\beq\label{UNdef}
U_N = \E^{i(A+iB)\frac{\Time}{N}}
\eeq
and using the following Grassmann symbols
\beq
\begin{split}
\cG(C_N^\dag D_N^\dag)(\bar{\psi},\psi) 
&= (1+i \frac{\Time}{N}V((U_N^\dag)^{\top} \bar{\psi},\psi))\E^{(\bar{\psi},U_N^\dag \psi)_{\Xpos}} \\
\cG(\rho_0)(\bar{\psi},\psi) 
&= 
\E^{(\bar{\psi},\E^{-\beta Q} \psi)_{\Xpos}}\\
\cG(C_N D_N)(\bar{\psi},\psi) 
&= 
(1-i \frac{\Time}{N}V(U_N^{\top} \bar{\psi},\psi))\E^{(\bar{\psi},U_N \psi)_{\Xpos}}\\
\cG \left(\E^{(c^+,a^*)_{\Xpos}} \E^{(c^-,a)_{\Xpos}}\right)(\bar{\psi},\psi) 
&= 
\E^{(c^+,\bar{\psi})_{\Xpos}} \E^{(c^-,\psi)_{\Xpos}} \E^{(\bar{\psi},\psi)_{\Xpos}}\;,
\end{split}
\eeq
one may thus show that
\begin{align}\label{eq:explicit_form}
\begin{aligned}
	Z_N(c^-,c^+) = &\int \prod_{\ell=1}^{N+1} \left[D_{{\Xpos}}(\bar{\psi}^-_\ell,\psi^-_\ell) \right] \prod_{\ell=2}^{N+1}\left[ (1+i \frac{\Time}{N}V((U_N^\dag)^{\top} \bar{\psi}_{\ell-1}^-,\psi^-_\ell))\E^{(\bar{\psi}^-_{\ell-1},U_N^\dag\psi^-_\ell)_{\Xpos}}\E^{-(\bar{\psi}^-_\ell,\psi^-_\ell)_{\Xpos}} \right]\\
	&\times  \int \prod_{k=1}^{N+1} \left[ D_{\Xpos}(\bar{\psi}_k^+,\psi_k^+)  \right] \E^{(c^+,\bar{\psi}^-_{N+1})_{\Xpos}} \E^{(c^-,\psi^+_{N+1})_{\Xpos}} \E^{(\bar{\psi}^-_{N+1},\psi^+_{N+1})_{\Xpos}} \E^{-(\bar{\psi}^+_{N+1},\psi^+_{N+1})_{\Xpos}}\\
	&\times \prod_{k=1}^N \left[ (1-i \frac{\Time}{N}V(U_N^{\top} \bar{\psi}^+_{k+1},\psi^+_k))\E^{(\bar{\psi}^+_{k+1},    U_N \psi^+_k)_{\Xpos}} \E^{-(\bar{\psi}^+_k,\psi^+_k)_{\Xpos}} \right] \E^{-(\bar{\psi}^+_{1},\Gamma \psi^-_1)_{\Xpos}}.
\end{aligned}
\end{align}
Here, in order to connect back to physics, we have labelled our Grassmann variables according to their respective times and appearance on the forward or backward path of the Keldysh contour. An illustration of this is given in Figure \ref{fig:sketch}.\\
\begin{figure}[h!]
\begin{center}
	\begin{tikzpicture}
		\draw[->][very thick] (-2,0) -- (12,0) node[anchor = west]{$t$};
		\draw[thick] (0,1) -- (10,1);
		\draw[thick] (0,-1) -- (10,-1);
		\draw[thick] (10,1) arc (90:-90:0.7cm and 1cm);
		\draw[thick] (0,1) arc (90:270:0.7cm and 1cm);
		\draw[thick][->] (4.99999,-1) -- (5,-1);
		\draw[thick][<-] (4.99999,1) -- (5,1);
		\filldraw (-0.7,0) circle (0.05) node[anchor=south east]{$\rho_0$};
		\filldraw (0.75,1) circle (0.05) node[anchor=north]{$\left( C_N D_N\right)^{\dag}$};
		\filldraw (2.25,1) circle (0.05) node[anchor=north]{$\left( C_N D_N\right)^{\dag}$};
		\filldraw (9.25,1) circle (0.05) node[anchor=north]{$\left( C_N D_N\right)^{\dag}$};
		\filldraw (7.75,1) circle (0.05) node[anchor=north]{$\left( C_N D_N\right)^{\dag}$};
		
		\filldraw (0.75,-1) circle (0.05) node[anchor=south]{$C_N D_N$};
		\filldraw (2.25,-1) circle (0.05) node[anchor=south]{$C_N D_N$};
		\filldraw (9.25,-1) circle (0.05) node[anchor=south]{$ C_N D_N$};
		\filldraw (7.75,-1) circle (0.05) node[anchor=south]{$C_N D_N$};
		
		\draw[thick] (0,-1) -- (0,-1.2) node[anchor=north]{$t_1^+ \leftrightarrow \bar{\psi}_1^+,\psi_1^+$};
		\draw[thick] (1.5,-1) -- (1.5,-1.2) node[anchor=north]{$t_2^+$};
		\draw[thick] (3,-1) -- (3,-1.2) node[anchor=north]{$t_3^+$};
		\draw[thick] (7,-1) -- (7,-1.2) node[anchor=north]{$t_{N-1}^+$};
		\draw[thick] (8.5,-1) -- (8.5,-1.2) node[anchor=north]{$t_{N}^+$};
		\draw[thick] (10,-1) -- (10,-1.2) node[anchor=north]{$t_{N+1}^+$};
		\filldraw (10.7,0) circle (0.05) node[anchor=south east]{$O$};
		\draw[thick] (0,1) -- (0,1.2) node[anchor=south]{$t_1^- \leftrightarrow \bar{\psi}_1^-,\psi_{1}^-$};
		\draw[thick] (1.5,1) -- (1.5,1.2) node[anchor=south]{$t_2^-$};
		\draw[thick] (3,1) -- (3,1.2) node[anchor=south]{$t_3^-$};
		\draw[thick] (7,1) -- (7,1.2) node[anchor=south]{$t_{N-1}^-$};
		\draw[thick] (8.5,1) -- (8.5,1.2) node[anchor=south]{$t_{N}^-$};
		\draw[thick] (10,1) -- (10,1.2) node[anchor=south]{$t_{N+1}^-$};
	\end{tikzpicture}	
	\end{center}
	\caption{A sketch of the operators inside the trace and their associated times. We start with our initial distribution $\rho_0$, go $N$ discrete time steps forward, insert $O$ - represented by the Grassmann sources $c^+,c^-$ in the actual calculation - and go $N$ discrete steps backwards to $\rho_0$. The associated Grassmann symbols take the $\bar{\psi}$ from the next and the $\psi$ from the previous time index along the Keldysh contour. The contour reflects the cyclicity of the trace.}	
	\label{fig:sketch}
\end{figure}
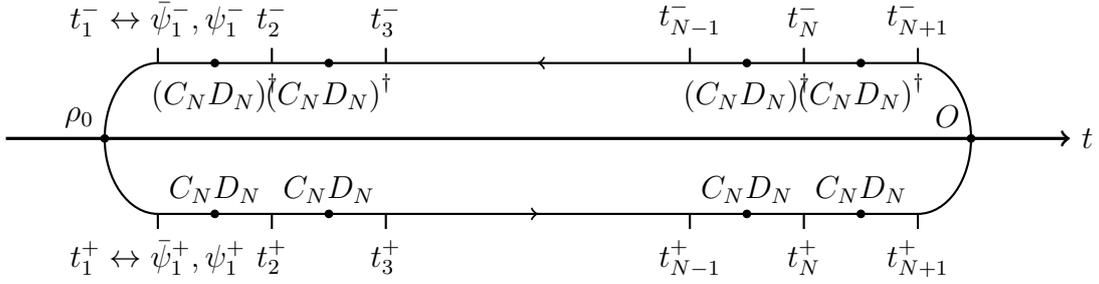

For $N\in \N$ let $\X^{(N)} = \{+,-\} \times \{1,\dots ,N+1\}\times {\Xpos} $. For  $(\rho,m);(\sigma,n)\in \{+,- \} \times \{1, \dots, N+1\}$ we define an $\X^{(N)} \times \X^{(N)}$ matrix (with operator-valued entries)
\begin{align}\label{eq:pre_mat}
\begin{aligned}
	\left(G_{\rho \sigma}^{(N)}\right)^{-1}_{m n} = 
	&\delta_{\rho,\sigma} \delta_{m,n} 
	- \delta_{\rho,-} \delta_{-,\sigma} \delta_{m,n-1} U_N^\dag- \delta_{\rho,+} \delta_{+,\sigma} \delta_{m-1,n} U_N \\
	& - \delta_{\rho,-} \delta_{+,\sigma} \delta_{m,N+1} \delta_{n,N+1} + \delta_{\rho,+} \delta_{-,\sigma} \delta_{m,1} \delta_{n,1} \E^{-\beta Q} \;.
\end{aligned}
\end{align}
This enables us to condense \eqref{eq:explicit_form} to
\begin{align}\label{eq:gen_fnct_G_inverse}
\begin{aligned}
	Z_N(c^-,c^+)=& \int \prod_{\ell=1}^{N+1} D_{{\Xpos}}(\bar{\psi}^-_\ell,\psi^-_\ell)  \int \prod_{k=1}^{N+1} D_{\Xpos}(\bar{\psi}_k^+,\psi_k^+) \E^{-(\bar{\psi},\left(G^{(N)}\right)^{-1} \psi)_{\X^{(N)}}} \E^{(c^+,\bar{\psi}^-_{N+1})_{\Xpos}} \E^{(c^-,\psi^+_{N+1})_{\Xpos}} \\
	&\times \prod_{\ell=2}^{N+1}(1+i \frac{\Time}{N}V((U_N^\dag)^{\top} \bar{\psi}^-_{\ell-1},\psi^-_\ell)) \prod_{k=1}^N (1-i \frac{\Time}{N}V(U_N^{\top} \bar{\psi}^+_{k+1},\psi^+_k)).
\end{aligned}
\end{align}
That $\left(G^{(N)}\right)^{-1}$ can be viewed as an actual inverse follows, as 
\beq
\det(\left(G^{(N)}\right)^{-1}) = \det(1+ \E^{-\beta Q}\E^{iA_+ \Time} \E^{-i A_- \Time})
\eeq
is nonvanishing. The scalar product involving $\left(G^{(N)}\right)^{-1}$ is short for
\begin{align}
\begin{aligned}
	(\bar{\psi},\left(G^{(N)}\right)^{-1} \psi)_{\X^{(N)}} &= \sum_{\rho ,\sigma} \sum_{m,n=1}^{N+1} (\bar{\psi}^{\rho}_m, \left(G_{\rho \sigma}^{(N)}\right)^{-1}_{m n}  \psi^{\sigma}_n)_{\Xpos}.
	\end{aligned}
\end{align}
We wish to express \eqref{eq:gen_fnct_G_inverse} via Grassmann Gaussian convolution to the covariance $G^{(N)}$, defined by
\begin{align}
\begin{aligned}
\mu_{G^{(N)}} \ast f(\bar{\eta},\eta) &= \gint{G^{(N)}} f(\bar{\psi}+\bar{\eta},\psi+\eta)\\
 &= \det(G^{(N)}) \int D_{\Xpos}(\bar{\psi},\psi) \ \E^{-(\bar{\psi}, \left(G^{(N)} \right)^{(-1)} \psi)_{\X^{(N)}}}f(\bar{\psi}+\bar{\eta},\psi+\eta).
 \end{aligned}
\end{align}
Thus, in order to express everything as a Grassmann Gaussian convolution, we should first establish the associated time discretized covariance $G^{(N)}$. The calculations for this are similar to the ones given in the appendices of \cite{Salmhofer:2009wm,Salmhofer:2020}. It is easiest to state $G^{(N)}$ in terms of its continuum limit $C$. 

\begin{definition}
Let $A$ and $B$ be the one-particle operators defining $H_0$ (see \eqref{H0AB}) and $A_\pm = A \pm \I B$, and $f_\beta (E) = (1+ \E^{\beta E})^{-1}$. Let $\X=\left(\{+,-\} \times [0,\Time ] \times {\Xpos}\right)$. The general Keldysh covariance $C$ is the $\X \times \X$-matrix given in the block form
	\begin{align}\label{eq:C}
	&C= \begin{pmatrix}
		C_{++} & C_{+-} \\
		C_{-+} & C_{--}
	\end{pmatrix}.
\end{align}
where each matrix element is a linear operator on the one-particle Hilbert space that depends on $s,s' \in [0,\Time ]$ as follows:
\beq\label{eq:general_covariance}
\begin{split}
C_{-+}(t,t')
&= 
\E^{iA_+(\Time-t)}\E^{-iA_- \Time} f_{-\beta} \E^{iA_-t'}, 
\\
C_{+-}(t,t') 
&= 
-\E^{-iA_-t} f_{-\beta} \E^{-\beta Q} \E^{iA_+t'},
\\
C_{++}(t,t') 
&= 
1_{t \geq t'} \E^{-iA_-t} f_{-\beta}  \E^{iA_-t'} +1_{t<t'} - \E^{-iA_-t} f_{\beta} \E^{iA_-t'},
\\
C_{--}(t,t') 
&= 
1_{t \leq t'} \E^{iA_+(\Time-t)} \E^{-iA_-\Time} f_{-\beta}  \E^{iA_- \Time} \E^{iA_+(\Time-t')} +1_{t>t'} \E^{iA_+ (\Time-t)} \E^{-i A_- \Time} f_{-\beta} \E^{-\beta Q} \E^{iA_+t'},
\end{split}
\eeq
where for notational brevity we used the indicator function $1_\mathcal{A}$ for some set $\mathcal{A}$ and in analogy to the normal Fermi distribution we wrote 
\begin{align}
f_{-\beta} =\left(1+\E^{-\beta Q} \E^{iA_+ \Time}\E^{-iA_- \Time} \right)^{-1}
\quad \quad f_{\beta} = \left(1+\E^{-\beta Q} \E^{iA_+ \Time}\E^{-iA_- \Time} \right)^{-1} \E^{-\beta Q} \E^{iA_+ \Time}\E^{-i A_- \Time}.
\end{align}
\end{definition}

We note in passing that one could also have arranged the operators in the products so as to put the Fermi functions between the operators $\E^{\pm \I A_\pm}$; this of course leads to the same results for the covariance. 

\begin{lemma}\label{lem:det_const}
For $N\in \N$ and defining $t_m = \frac{\Time}{N}(m-1)$ the time discretized covariance is
\begin{align}\label{eq:connection_G_C}
G^{(N)}\left((\rho,m,x),(\sigma,n,y)\right) = C\left((\rho,t_m,x),(\sigma,t_n,y)\right).
\end{align}
Let $B_1^{(n)}$ denote the closed unit ball on $\C^{n}$, with usual scalar product. Then for $n \in \N$,  and $X_1, \dots , X_n, Y_1, \dots , Y_n \in \X$, there is a finite determinant bound $\delta_C$ (see \cite[Section~1]{Pedra2008} for more information), such that 
	\begin{align}
		\sup_{v_1, \dots , v_n, q_1, \dots , q_N \in B_1^{(n)}} \abs{\det \left[\left( \bracket{v_i}{q_j} C(X_i, Y_j) \right)_{1\leq i,j \leq n} \right]} \leq \delta_C^{2n}.
	\end{align}
By \eqref{eq:connection_G_C} $\delta_C$ is a determinant bound for all time discretized covariances as well.
\end{lemma}
\begin{proof}
	That $G^{(N)}\left((\rho,m,y),(\sigma,n,y)\right) = C\left((\rho,t_m,y),(\sigma,t_n,y)\right)$ actually corresponds to the time discretized covariance can be checked explicitly. In actuality, one would always derive $G^{(N)}$ by explicitly inverting \eqref{eq:pre_mat} and simply define $C$ as its continuum limit.
	
	 The existence of some finite determinant bound for the covariance given in \eqref{eq:general_covariance} is easily proven by finding a suitable Gram representation and using \cite[Theorem~1.3]{Pedra2008}. An explicit version of this argument is used in the proof of Lemma \ref{lem:determinant_bound_and_decay_const}. It is clear that, by definition, this determinant bound holds for all $G^{(N)}$.
\end{proof}

Using that the factors of $U_N$ in the argument of $V$ of \eqref{eq:gen_fnct_G_inverse} drop out for $N \rightarrow \infty$, and completing the square, we get 
\begin{align}
\begin{aligned}
Z_N(c^-,c^+)=& \det(1+ \E^{-\beta Q}\E^{iA_+ \Time} \E^{-i A_- \Time}) \E^{-(c^-,\left(G_{+-}^{(N)}\right)_{N+1,N+1}c^+)_{\Xpos}}\\
	& \times\mu_{{G}^{(N)}} \ast \nu^{(N)} (\bar{\zeta}^{(N)} , \zeta^{(N)})
	 \end{aligned}
\end{align}
with 
\begin{align}
	\nu^{(N)} = \prod_{k=1}^N (1-i \frac{\Time}{N}V)(\bar{\psi}_{k+1}^+,\psi_{k}^+) \prod_{\ell=2}^{N+1}(1+i \frac{\Time}{N}V^{\dag}(\bar{\psi}_{l-1}^-,\psi_l^-))
\end{align}
and 
\begin{align}
	\bar{\zeta}^{(N)} = \left(G_{\cdot,+}^{(N)}\right)_{\cdot,N+1}^{\top} c^-, \quad \quad \quad \quad \zeta^{(N)}  = -\left(G^{(N)}_{\cdot,-} \right)_{\cdot,N+1}c^+.
\end{align}
The part of the covariance contributing to the reduced density matrices even for zero potential is
\begin{align}\label{eq:equiv}
 \left(G_{+-}^{(N)}\right)_{N+1,N+1} &= -U_N^N \left(1+\E^{-\beta Q} \left(U_N^{\dag}\right)^NU_N^N\right)^{-1} \E^{-\beta Q} \left(U_N^{\dag}\right)^NU_N^N \left(U_N^{-1}\right)^N.
\end{align}
Besides giving an explicit form of $G^{(N)}$ through its continuum limit, Lemma \ref{lem:det_const} states the existence of a decay constant $\delta_C$ applicable to all time discretized covariances. This allows us to re-exponentiate all terms of $\nu^{(N)}$.To see so, we begin by defining
\begin{align}
\V^{(N)} = \frac{\Time}{N}\left(\sum_{k=1}^{N} V(\bar{\psi}^+_{k+1},\psi^+_k) - \sum_{\ell=2}^{N+1} V^{\dag}(\bar{\psi}^-_{\ell-1},\psi^-_\ell) \right).
\end{align}
\begin{lemma}
	Let $\nu^{(N)}$ and $\V^{(N)}$ be defined as above. We define 
\begin{align}
	\Delta = \E^{-i\V^{(N)}}-\nu^{(N)}
\end{align}
Let $q>0$. Given any norm $\abs{\cdot}$ on $\C\{c^+(x),c^-(x)\}_{x \in {\Xpos}}$ and 
\beq
f=\sum_{m,\bar{m}} \int_{\X^{(N)}} \rd^m X  \rd^{\bar{m}} Y f_{m,\bar{m}}(X,Y) \bar{\psi}^{\bar{m}}(Y) \psi^m(X)
\eeq 
with $f_{m,\bar{m}}$ potentially dependent on $c^{\pm}$, we define the submultiplicative norm
\begin{align}\label{eq:Grassmann_Algebra_Norm}
	\mynorm{f}_q = \sum_{m,\bar{m}} \int_{\X^{(N)}} \rd^m X \int_{\X^{(N)}} \rd^{\bar{m}} Y |f_{m,\bar{m}}(X,Y)| q^{m+\bar{m}}.
\end{align}
Then
\begin{align}
	\begin{aligned}
		 \abs{\mu_{G^{(N)}} \ast \Delta(\bar{\zeta},\zeta)} &\leq \mynorm{\Delta}_{\delta_C} \leq 2\frac{T^2}{N} \mynorm{V}_{\delta_C}^2 \E^{2 \Time \mynorm{V}_{\delta_C}} \rightarrow 0,
	\end{aligned}
\end{align}
where the limit $N\rightarrow \infty$ is taken in the end.
\end{lemma}
\begin{proof}
This proof is essentially the same as in \cite{Salmhofer:2009wm}. We include it here because it shows how simple error estimates become upon using the norms we have introduced on the Grassmann algebra. The intuition is that every power of the Grassmann fields gets replaced by a power of $\delta_C$, i.e.\ the determinant constant can be regarded as the natural `size' of each Grassmann-valued field in the Grassmann Gaussian integral, similarly to the way one can think of the standard deviation as the `typical' size of a real-valued Gaussian variable that is centered at zero.

The first inequality follows by \cite[Lemma~B.7]{MSbook} combined with the definition of the determinant bound as in \cite[Definition~1.2]{Pedra2008}. The second inequality may be derived similarly to \cite[Equation~(57)]{Salmhofer:2009wm} using submultiplicativity of the norm and keeping in mind that hermitian conjugation is an isometry of $\mynorm{\cdot}_{\delta_C}$. As $V$ is local in time, there is no dependence of $\mynorm{V}_{\delta_C}$ on $N$ and so the limit of the bound vanishes as $N \to \infty$. We say more about the implications of locality in Remark \ref{rem:locality_V}.
\end{proof}

Thus we have proven the following theorem, where convergence holds in any norm on $\C\{c^+(x),c^-(x)\}_{x \in {\Xpos}}$.  

\begin{theorem}\label{theorem-FI}
The generating function $Z(c^-,c^+)$ is the limit 
\begin{align}\label{eq:path_integral_Z}
\begin{aligned}
	Z(c^-,c^+)=& \det(1+ \E^{-\beta Q } \E^{i A_+ \Time}\E^{- i A_- \Time}) \E^{(c^-,\E^{-i A_- \Time} f_{\beta} \E^{i A_- \Time}c^+)_{\Xpos}}\\
	 & \times \lim_{N \rightarrow \infty} \left( \mu_{G^{(N)}} \ast \E^{-i \V^{(N)}}(\bar{\zeta}^{(N)},\zeta^{(N)})\right).
\end{aligned}
\end{align}
\end{theorem}

The key point in the proof was the existence of a continuum covariance with determinant bound $\delta_C$. Let us conclude with two remarks.

\begin{remark} \label{rem:locality_V}
In the limit $N \rightarrow \infty$, $\V= \lim_{N\rightarrow \infty} \V^{(N)}$ is formally given by
	\begin{align}
		\V &=  \int_{0}^{\Time} \rd t \  \left(V(\bar{\psi}^+_t,\psi^+_t)) - V^{\dag}(\bar{\psi}^-_t,\psi^-_t)\right).
	\end{align}
Thus it is local in both time indices $(t,\sigma)$. This locality reflects itself in the norm bounds on truncated expectation values in Theorem \ref{th:main_theorem}. The Theorem uses a slightly different norm than the norm given in \eqref{eq:Grassmann_Algebra_Norm}, namely it uses a norm directly on the coefficient functions $g_{m,\bar{m}}: \X^{(N)} \times \X^{(N)} \rightarrow \C$. We define
$\abs{g_{m,\bar{m}}}_{1,\infty , \X^{(N)}}  = \max\{\gamma_1, \gamma_2\}$ where
\beq\begin{split}
\gamma_1 
&= \max_j \sup_{X_j \in \X^{(N)}} \int_{\X^{(N)}} \prod_{\substack{i=1 \\ i\neq j}}^m \rd X_i \rd^{\bar{m}}Y \abs{g_{m,\bar{m}}(X,Y)} \; ,
\\
\gamma_2 
&= \max_j \sup_{Y_j \in \X^{(N)}} \int_{\X^{(N)}}  \rd^mX\prod_{\substack{i=1 \\ i\neq j}}^{\bar{m}} \rd Y_i \abs{g_{m,\bar{m}}(X,Y)}   \; .
\end{split}
\eeq
and
\begin{align}\label{eq:old_seminorm}
	\eunorm{g}_{h,\X^{(N)}} = \sum_{\substack{m,\bar{m} \geq 0}} |g_{m,\bar{m}}|_{1,\infty , \X^{(N)}} h^{m+\bar{m}}.
\end{align}
Locality of $\V$ and each individual $\V^{(N)}$ then translates to $\eunorm{\V^{(N)}}_{h,\X^{(N)}} = \eunorm{V}_{h,\Xpos}$. 
\end{remark}
\begin{remark}\label{rem:concrete_scenarios}
	\eqref{eq:general_covariance} and \eqref{eq:path_integral_Z} are rather complicated to work with directly. Thus, the explicit bounds on truncated expectation values given below in Lemma \ref{sec:example} focus on two simpler cases in which the covariance simplifies drastically. In both cases we assume $[A,B]=0$ and $[B,Q]=0$, such that
\begin{align}\label{eq:covariance_A_B_commute}
		\begin{aligned}
	C_{-+}(t,t') =& \E^{-it A} f_{\beta} \left(-(Q-2 \frac{\Time}{\beta}B) \right) \E^{it' A} \E^{(-2\Time+t+t')B},\\
	C_{+-}(t,t') =& -\E^{-itA} f_{\beta}\left(-(Q-2 \frac{\Time}{\beta}B) \right) \E^{-\beta Q} \E^{it'A}\E^{-(t'+t)B}, \\
	C_{++}(t,t') =& 1_{t \geq t'}\E^{-it A} f_{\beta}\left(-(Q-2 \frac{\Time}{\beta}B) \right) \E^{it' A}\E^{(t'-t)B}\\
	& -1_{t<t'} \E^{-itA}f_{\beta} \left(-(Q-2 \frac{\Time}{\beta}B) \right) \E^{-\beta Q}  \E^{it'A}\E^{(t'-t-2\Time) B}, \\
	C_{--}(t,t') =& 1_{t' \geq t}\E^{-it A} f_{\beta}\left(-(Q-2 \frac{\Time}{\beta}B) \right) \E^{it' A} \E^{(t-t')B}\\
	&-1_{t'<t} \E^{-itA} f_{\beta}\left(-(Q-2 \frac{\Time}{\beta}B) \right) \E^{-\beta Q}  \E^{it'A}\E^{(t-t'-2\Time) B}.
		\end{aligned}
\end{align}
	Here we used the usual Fermi function $f_{\beta}(E)=\left(1+\E^{\beta E} \right)^{-1}$. We denote the basis in which both $Q$ and $B$ are simultaneously diagonalized by $\mathscr{L}$. We write $q_\ell$ and $b_\ell$ for the respective eigenvalues. Moreover, we define 
\begin{align}\label{eq:gamma_ell}
\gamma_\ell = q_\ell + 2\frac{\Time}{\beta}b_\ell
\end{align} for the eigenvalues occurring in the Fermi function. For the explict calculations in Section \ref{sec:example} we will look at the case of a truly dissipative quantum systems, i.e. $b_{\ell}>0$ for all $\ell \in \mathscr{L}$ and the case of unitary time evolution, i.e. B=0. In the unitary case the above covariance simplifies to the usual Keldysh covariance
\beq\label{eq:C_for_B=0}
\begin{split}
C_{-+}(t,t') 
&= 
\E^{-it A} f_{\beta}(-Q) \E^{it' A} 
\\
C_{--}(t,t') 
&= 
1_{t'<t}C_{+-}(t,t') + 1_{t' \geq t} C_{-+}(t,t')
\\
C_{+-}(t,t') 
&= 
-\E^{-it A} f_{\beta}(Q) \E^{it' A} 	
\\
C_{++}(t,t') 
&= 
1_{t<t'}C_{+-}(t,t') + 1_{t \geq t'} C_{-+}(t,t') 
\end{split}
\eeq
Furthermore, the generating functional becomes
\begin{align}\label{eq:final_form}
\begin{aligned}
	Z(c^-,c^+)= \det(1+ \E^{-\beta Q}) \E^{(c^-,\E^{-i \Time A} f_{\beta}(Q) \E^{i \Time A} c^+)_{\Xpos}} \lim_{N \rightarrow \infty} \left( \mu_{G^{(N)}} \ast \E^{-i \V^{(N)}}(\bar{\zeta}^{(N)},\zeta^{(N)})\right).
\end{aligned}
\end{align}
As expected, there is no contribution to the connected $n$-point functions for the case of the free theory ( $V=0$ ) for $n \ge 4$, and the free connected two-point function is given by $\E^{-i \Time A} f_{\beta}(Q) \E^{i \Time A}$ and thus represents the time evolution of the Fermi-Dirac distribution derived from von Neumann's equation. 
\end{remark}

\subsection{Bounds on the Truncated Expectation Values}
If we insert our path integral formulation into the definition of the generating functional of truncated expectation values \eqref{eq:def_F} and define the Wilsonian effective action
\begin{align} \label{eq:def_effective_action}
\W^{(N)}(G^{(N)},V) = - \log (\frac{1}{\left( \mu_{G^{(N)}} \ast \E^{-\V}(0,0)\right)}\mu_{G^{(N)}} \ast \E^{-\i V^{(N)}}(\bar{\zeta}^{(N)},\zeta^{(N)})),
\end{align}
we see that
\begin{align}\label{eq:expansion_F0}
	\begin{aligned}
		F(c^-,c^+) =(c^-,\E^{-i \Time A_-} f_{\beta} \E^{i \Time A_-} c^+)_{\Xpos} - \lim_{N \rightarrow \infty}\W^{(N)}.
			\end{aligned}
\end{align}
We are now in a positon to formulate the bounds on the truncated expectation values. These bounds contain two more characteristic data of the covariance
\begin{definition}\label{def:decay_const}
We define the decay constant 
\begin{align}
	\begin{aligned}
		\alpha_C &= \max \left\{\sup _{X \in \X} \int_\X|C(X, Y)| \mathrm{d} Y, \sup _{X \in \X} \int_\X |C(Y, X)| \mathrm{d} Y\right\}= \eunorm{C}_{1,\infty,\X}
		\end{aligned}
	\end{align}
and the modified decay constant  
	\begin{align}
		\tilde{\alpha}_C = \max\left\{\sup_{Z\in\X} \int_{\Xpos} \rd x \ \abs{C\left[(\Time,+,x),Z \right]},\sup_{Z\in\X} \int_{\Xpos} \rd x \ \abs{C\left[Z,(\Time,-,x) \right]}  \right\}.
	\end{align}
\end{definition}
Clearly, $\alpha_C$ is finite for the general covariance \eqref{eq:general_covariance}. Explicit bounds for the (modified) decay constant in the setting of Remark \ref{rem:concrete_scenarios} will be given in Section \ref{sec:example}.

We can now give the proof of Theorem \ref{th:main_theorem}.

\begin{proof}
	As $V$ is even and, due to locality of $\V$, $\omega_C \eunorm{\V}_{3\delta_C,\X} \leq \frac{1}{2}$ all conditions of \cite[Theorem~1]{Salmhofer:2009wm} -- replacing Gram constants by determinant bounds when adequate -- are fulfilled. Thus the expansion of $\W^{(N)}$ in terms of powers of $\V^{(N)}$ converges and we have analyticity of the limit $\W(C,V)=\lim_{N\rightarrow \infty} \W(G^{(N)},V)$ in the fields $(\bar{\zeta},\zeta)$.
	
As the tree expansion for $W_{m,\bar{m}}^{(N)}=\frac{1}{m! \bar{m}!} \frac{\delta^{m}}{\delta \psi^{m}}  \frac{\delta^{\bar{m}}}{\delta \bar{\psi}^{\bar{m}}} \W^{(N)}$ given in \cite{Salmhofer:2009wm} is  convergent in the limit $N \rightarrow \infty$, we can write $W_{m,\bar{m}} = \lim_{N\rightarrow \infty} W^{(N)}_{m,\bar{m}}$,  and the generating function is given by
\begin{align}\label{eq:expansion_F}
	\begin{aligned}
		F(c^-,c^+) =(c^-,\E^{-i \Time A_-} f_{\beta} \E^{i \Time A_-} c^+)_{\Xpos} - \sum_{\substack{m,\bar{m}\\ m+\bar{m} \text{ even}}} \int_\X \rd^m X  \rd^{\bar{m}}Y  W_{m,\bar{m}}(X,Y) \bar{\zeta}^{\bar{m}}(Y) \zeta^{m}(X).
	\end{aligned}
\end{align}
Applying the appropriate derivatives to \eqref{eq:expansion_F} we may write the truncated expectation values as
\begin{align}\label{aliali}
	\begin{aligned}
	 \gammaT_{m,\bar{m}} &(x_1,\dots,x_m,y_1,\dots ,y_{\bar{m}},\Time)  = \E^{-i \Time A_-} f_{\beta}(Q)_{y_1,x_1} \E^{i \Time A_-} \delta_{m,1}\delta_{\bar{m},1} \\
	 &+ (-1)^{m-1} \int_{\X} \prod_{k=1}^m \rd Z_k  \int_{\X} \prod_{\ell=1}^{\bar{m}}  \rd Z'_\ell  \ W_{m,\bar{m}}\left( Z_1, \dots, Z_m, Z'_1, \dots,Z'_{\bar{m}} \right)\\
	 &\times m! \bar{m}! \prod_{\ell=1}^{\bar{m}} C\left((+,\Time,y_\ell),Z'_\ell \right) \prod_{k=1}^m C\left(Z_k,(-,\Time,x_k) \right).
	\end{aligned}
	\end{align}
	While the expression may seem unintuitive at first, Figure \ref{fig:sketch_truncated_exp_values} gives a nice graphical representation of the truncated expectation values for $m\neq 1$ or $\bar{m} \neq 1$.
	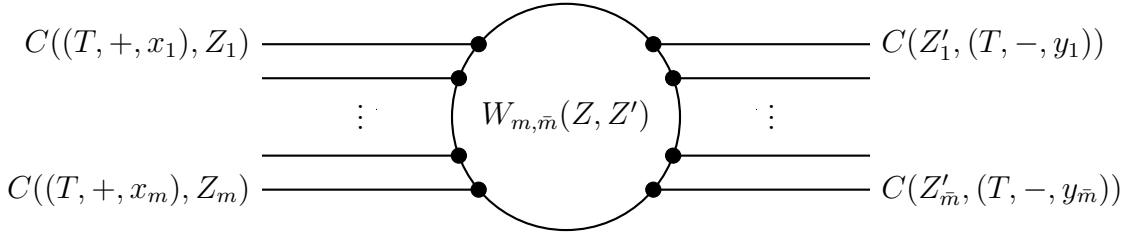
\begin{figure}[h!]
\begin{center}
	\begin{tikzpicture}
		\draw[thick] (0,0) circle (1.5cm)node{$W_{m,\bar{m}}(Z,Z')$};
		\filldraw (20:1.5) circle (0.1cm);
		\filldraw (40:1.5) circle (0.1cm);
		\filldraw (-20:1.5) circle (0.1cm);
		\filldraw (-40:1.5) circle (0.1cm);
		\draw[thick] (20:1.5cm) -- (4,{1.5* sin(20)});
		\draw[thick] (40:1.5cm) -- (4,{1.5* sin(40)})node[anchor = west]{$C(Z'_1,(\Time,-,y_1))$};
		\draw[thick] (-20:1.5cm) -- (4,{1.5* sin(-20)});
		\draw[thick] (-40:1.5cm) -- (4,{1.5* sin(-40)})node[anchor = west]{$C(Z'_{\bar{m}},(\Time,-,y_{\bar{m}}))$};
		
		\draw (2.5,0.1) -- (2.5,0.1)node[anchor=west]{\vdots} ;
		\draw (-2.5,0.1) -- (-2.5,0.1)node[anchor=east]{\vdots} ;
		
		\filldraw (1600:1.5) circle (0.1cm);
		\filldraw (140:1.5) circle (0.1cm);
		\filldraw (200:1.5) circle (0.1cm);
		\filldraw (220:1.5) circle (0.1cm);
		\draw[thick] (200:1.5cm) -- (-4,{1.5* sin(-20)});
		\draw[thick] (220:1.5cm) -- (-4,{1.5* sin(-40)})node[anchor = east]{$C((\Time,+,x_m),Z_m)$};
		\draw[thick] (160:1.5cm) -- (-4,{1.5* sin(20)});
		\draw[thick] (140:1.5cm) -- (-4,{1.5* sin(40)})node[anchor = east]{$C((\Time,+,x_1),Z_1)$};
		
	\end{tikzpicture}	
	\end{center}
	\caption{A graphical representation of the truncated expectation values for $m\neq1$ or $\bar{m}\neq 1$. The large dots stand for integration over the $Z$ and $Z'$ variables in \eqref{aliali}.}
	\label{fig:sketch_truncated_exp_values}
\end{figure}

Taking the one-infinity norm we have
\begin{align*}
	\frac{1}{m! \bar{m}!} \abs{\gammaT_{m,\bar{m}}-\gammaT_{m,\bar{m}}\big|_{V=0}}_{1,\infty, {\Xpos}} =& \max_{i}\sup_{x_i \in {\Xpos}} \int \prod_{x_j \neq x_i} \rd x_j  \bigg|\int \prod_{Z_j} \rd Z_j W(Z_1, \dots Z_{m+\bar{m}}) \\
		&\times \prod_{\ell=1}^m C\left((\Time,+,x_\ell),Z_\ell\right) \prod_{k=1}^{\bar{m}} C\left(Z_{m+k},(\Time,-,x_{k+m}\ell)\right)\bigg| \\
  \leq&\eunorm{W_{m,\bar{m}}}_{1,\infty,\X}\alpha_C \tilde{\alpha}_C^{m+\bar{m}-1} \leq 2 \tilde{\alpha}_C^{m+\bar{m}-1} \alpha_C \delta^{-m-\bar{m}} \eunorm{V}_{3\delta_C,{\Xpos}}.
\end{align*}
The last inequality follows, as by \cite[Theorem~1]{Salmhofer:2009wm}  $\abs{W_{m,\bar{m}}}_{1,\infty,\X} \leq 2 \delta^{-m-\bar{m}} \eunorm{\V}_{3\delta_C,{\Xpos}}$.

The argument how these bounds imply bounds for unordered monomials of degree $n\neq2$ can be found (up to some minor sign errors) in \cite[Section~2.2]{Salmhofer:2009wm} and remains essentially unchanged.
\end{proof}

\begin{remark}
The derivation of the qualitative bounds on truncated expectation values is largely analogous to the derivation of \cite[Theorem~1]{Salmhofer:2009wm}, albeit in a slightly different context. However, our result improves upon \cite[Theorem~1]{Salmhofer:2009wm} by replacing the decay constant via the modified decay constant. This improvement translates to the context of \cite[Theorem~1]{Salmhofer:2009wm}.
\end{remark}

\subsection{Examples of Determinant Bounds and Decay Constants} \label{sec:example}

We see that the essential quantities for the bounds given in Theorem \ref{th:main_theorem}, determining size and even existence of the bounds, are the determinant bound $\delta_C$ and the (modified) decay constant $\alpha_C$. Henceforth, in order to give physical meaning to the bound, we need explicit bounds for $\alpha_C$ and $\delta_C$. We will do so explicitly for the two scenarios established in Remark \ref{rem:concrete_scenarios}. 

\begin{lemma}\label{lem:determinant_bound_and_decay_const}
Let $C$ be the covariance defined in \eqref{eq:covariance_A_B_commute}. For $[A,B]=0$ and $[Q,B]=0$ we have the following determinant bounds
\begin{alignat}{2}
	&B>0: \quad \delta_C=6\varepsilon^{-\frac{1}{2}}(1+\E^{-\frac{1}{2}\beta \tilde{q}}), \quad \quad \quad && B=0: \quad \delta_C=12\varepsilon^{-\frac{1}{2}},
\end{alignat}
where for any variable $\phi_\ell$, we define $\tilde{\phi}$ such that
\begin{align}
	 \tilde{\phi} = 	\inf_{\ell}  \phi_{\ell}. 
\end{align}
Moreover, we have the following general upper bounds for the decay constants
	\begin{alignat*}{2} 
	&B>0:  \quad && \quad \quad \quad B=0 \\
	& \alpha_C \leq 2 \abs{{\Xpos}} \left(1+\E^{-\beta \tilde{q}}\right) \frac{1}{\tilde{b}} \left(1-\E^{-\Time \tilde{b}} \right),\quad   &&  \quad \quad \quad \alpha_C \leq 2 \abs{{\Xpos}} \Time, \tag{\stepcounter{equation}\theequation}\\
	&  \tilde{\alpha}_C \leq \abs{{\Xpos}} \max\left\{1,\E^{-\beta \tilde q} \E^{-\Time \tilde{b}}\right\},   \quad && \quad  \quad \quad \tilde{\alpha}_C \leq \abs{{\Xpos}},
	\end{alignat*}
	where, as in \eqref{eq:gamma_ell}, $q_{\ell}$ and $b_{\ell}$ refer to the eigenvalues of $Q$ and $B$, respectively.
\end{lemma}

\begin{proof}
	Let us begin by showing the respective determinant bounds for $B>0$. The proof may be summarized as constructing Gram representation for the individual block matrices and combining these by \cite[Theorem~1.3]{Pedra2008}.\\
	The Gram representations are inspired by \cite[Section~4.1]{Pedra2008} and thus we define for $q \in \R \setminus \{0\}$ 
\begin{align}
	\phi(s,q) = \frac{1}{\sqrt{\pi}} \frac{\sqrt{q f_{\beta}(-q)}}{is-q}.
\end{align}
Without loss of generality we assume that $\gamma_\ell \neq 0$, where $\gamma_\ell$ was defined in \eqref{eq:gamma_ell}. Otherwise, apply a small offset $\alpha$ to $Q$. As the determinant is continuous one can then take the limit $\alpha \rightarrow 0$. We define
	\begin{align}
		g_{t,x}^{\sigma, \rho}(s,\ell) &= \bracket{\ell}{\E^{itA} x} \phi(s,\abs{\gamma_\ell}) \E^{\rho it \frac{b_\ell}{\abs{\gamma_\ell}}s}1_{\sigma \gamma_\ell >0}
	\end{align}
and $h^{\rho}_{t,x}(s,l)=g_{t,x}^{-\rho}(s,\ell)\E^{is\beta}$. Similarly, let $\tilde{g};\tilde{h}$ contain an additional factor of $\E^{-\frac{1}{2}\beta q_\ell}$ and let the capitalized versions contain an additional factor of $\E^{i2\Time \frac{b_\ell}{\abs{\gamma_\ell}}s}$. As $\eunorm{\phi(\cdot,q)}_2 = f_{\beta}(q) \leq 1$, all these variants are clearly in $\cH = L^2(\R \times \mathscr{L},\rd s \times \rd \mu)$, where $\mu$ denotes the counting measure on $\mathscr{L}$ and $s$ the usual Lebesque measure. Furthermore, their norms are  $\eunorm{H_{t,x}^{\rho}}=\eunorm{h_{t,x}^{\rho}}=\eunorm{G_{t,x}^{\rho,\sigma}}=\eunorm{g_{t,x}^{\rho,\sigma}}=\varepsilon^{-\frac{1}{2}}$ and 
	$\eunorm{\tilde{H}_{t,x}^{\rho}}=\eunorm{\tilde{h}_{t,x}^{\rho}}=\eunorm{\tilde{G}_{t,x}^{\rho,\sigma}}=\eunorm{\tilde{g}_{t,x}^{\rho,\sigma}}=\varepsilon^{-\frac{1}{2}}\E^{-\frac{1}{2}\beta \tilde{q}}$. 
Calculation gives
\beq
\begin{split}
C((+,t,x),(-,t',y)) &= \bracket{-\tilde{g}^{+-}_{t,x}-\tilde{g}^{--}_{t,x}}{\tilde{g}^{++}_{t',y}+\tilde{h}^{+}_{t',y}} ,\\
C((-,t,x),(+,t',y))&= \bracket{g_{t,x}^{++}+g_{t,x}^{-+}}{G_{t',y}^{+-}+H_{t',y}^{-}}, 
\end{split}
\eeq
and
\beq
\begin{split}
C((+,t,x),(+,t',y))&= 1_{t>t'} \bracket{-\tilde{g}^{+-}_{t,x}-\tilde{g}^{--}_{t,x}}{\tilde{G}^{+-}_{t',y}+\tilde{H}^{-}_{t',y}} +1_{t \leq t'} \bracket{g_{t,x}^{+-}+g_{t,x}^{--}}{g_{t',y}^{+-}+h_{t',y}^{-}},\\
C((-,t,x),(-,t',y)) &=1_{t>t'} \bracket{-\tilde{g}^{++}_{t,x}-\tilde{g}^{-+}_{t,x}}{\tilde{G}^{++}_{t',y}+\tilde{H}^{+}_{t',y}} +1_{t \leq t'} \bracket{g_{t,x}^{++}+g_{t,x}^{-+}}{g_{t',y}^{++}+h_{t',y}^{+}}.
\end{split}
\eeq
By \cite[Theorem~1.3]{Pedra2008} this implies
\begin{alignat*}{2}
	&\delta_{C_{+-}} = 2\varepsilon^{-\frac{1}{2}}\E^{-\frac{1}{2}\beta \tilde{q}},
	&&\delta_{C_{-+}} = 2\varepsilon^{-\frac{1}{2}}, \tag{\stepcounter{equation}\theequation}\\
	&\delta_{C_{++}} = 2\varepsilon^{-\frac{1}{2}}(1+\E^{-\frac{1}{2}\beta \tilde{q}}),
	\quad &&\delta_{C_{--}} = 2\varepsilon^{-\frac{1}{2}}(1+\E^{-\frac{1}{2}\beta \tilde{q}}).
\end{alignat*}
By an argument very similar to \cite[Lemma~3.11]{Pedra2008} the sum of the individual determinant bounds gives a determinant bound for the entire covariance matrix concluding the proof.

An almost equivalent proof is possible for $B=0$. One then defines $g_{t,x}^{\sigma}(s,\ell) = \bracket{\ell}{\E^{i \cE t} x}\phi(s,|q_\ell|) 1_{\sigma q_\ell > 0}$ and $h_{t,x}^{\sigma}(s,l)=g_{t,x}^{\sigma}(s,\ell)\E^{is\beta}$. Then
	\begin{align}
	\begin{aligned}
		C((+,t,x),(-,t',y))&= \bracket{-g_{t,x}^+ -g_{t,x}^- }{h_{t',y}^+ +g_{t',y}^-}, \\
		C((-,t,x),(+,t',y))&= \bracket{g_{t,x}^+ +g_{t,x}^- }{g_{t',y}^+ +h_{t',y}^-} .
	\end{aligned}
	\end{align}
	By the peculiar structure of the covariance in \eqref{eq:C_for_B=0} and as $\eunorm{h_{t,x}^\sigma}=\eunorm{g_{t,x}^\sigma} = 2\varepsilon^{-\frac{1}{2}}$the previous arguments imply $\delta_C = 12 \varepsilon^{-\frac{1}{2}}$. This concludes the proof for the determinant bounds.\\
	
Let us now focus on proving the bounds for the decay constants in the case of $B>0$. (The bound given in Lemma \ref{lem:determinant_bound_and_decay_const} for $B=0$ is trival.)
Without assuming additional structure on $Q$, we may bound the decay constant by first establishing a bound for the operator norm $\eunorm{C_{\rho \sigma}(t,t')}$. Using the representation given in \eqref{eq:covariance_A_B_commute} we get
\begin{align}\label{eq:crude_est}
	\begin{aligned}
		\eunorm{C_{+-}(t,t')} &\leq \E^{-\beta \tilde{q}} \E^{-(t'+t)\tilde{b}} ,\\
		\eunorm{C_{-+}(t,t')} &\leq \E^{(-2\Time+t'+t)\tilde{b}}, \\
		\eunorm{C_{++}(t,t')} &\leq 1_{t<t'} \E^{-\beta \tilde{q}} \E^{(-2\Time+t'-t)\tilde{b}} +1_{t\geq t'}\E^{(t'-t)\tilde{b}},\\
		\eunorm{C_{--}(t,t')} &\leq 1_{t'<t} \E^{-\beta \tilde{q}} \E^{(-2\Time+t-t')\tilde{b}} +1_{t'\geq t}\E^{(t-t')\tilde{b}}.
	\end{aligned}
\end{align}
Assuming $\tilde{b} \neq 0$ ony may calculate that
\begin{align}\label{eq:first_last}
	\begin{aligned}
	\sup_{t} \int \rd t' \abs{(C((+,t,x),(-t',y))} &\leq \sup_{t} \varepsilon^{-1} \frac{1}{\tilde{b}} \E^{-\beta \tilde{q}}(1-\E^{-\Time \tilde{b}})\E^{-t \tilde{b}} = \varepsilon^{-1} \E^{-\beta \tilde{q}} \frac{1}{\tilde{b}}(1-\E^{-\Time \tilde{b}}),  \\
	\sup_{t} \int \rd t' \abs{(C((-,t,x),(+,t',y))} &\leq \sup_{t} \varepsilon^{-1} \frac{1}{\tilde{b}}(1-\E^{-\Time \tilde{b}})\E^{-(\Time-t) \tilde{b}} = \varepsilon^{-1} \frac{1}{\tilde{b}}(1-\E^{-\Time \tilde{b}}), \\	
	\sup_{t} \int \rd t' \abs{(C((+,t,x),(+,t',y))} &\leq \sup_{t} \varepsilon^{-1} \frac{1}{\tilde{b}}(1+\E^{-\beta \tilde{q}}) \left( \E^{(-\Time-t)\tilde{b}}-\E^{-2\Time \tilde{b}} +1 - \E^{-t \tilde{b}} \right)\\
	& = \varepsilon^{-1} \left(1+\E^{-\beta \tilde{q}}\right) \frac{1}{\tilde{b}}(1-\E^{-\Time \tilde{b}}),\\
	\sup_{t'} \int \rd t \abs{(C((+,t,x),(+,t',y))} &\leq \sup_{t'} \varepsilon^{-1} \frac{1}{\tilde{b}}(1+\E^{-\beta \tilde{q}}) \left( \E^{(-2\Time+t')\tilde{b}}-\E^{-2\Time \tilde{b}} +1 - \E^{(t'-\Time) \tilde{b}} \right)\\
	& = \varepsilon^{-1} \left(1+\E^{-\beta \tilde{q}}\right) \frac{1}{\tilde{b}}(1-\E^{-\Time \tilde{b}}). 
\end{aligned}
\end{align}
As $\eunorm{C_{\pm \mp}}$ is symmetric under exchange of $t$ and $t'$ and $\eunorm{C_{\pm \pm}}$ differ simply by exchanging $t$ and $t'$  the decay constant is then bounded by
\begin{align}
	\begin{aligned}
		\alpha_C &= \left.\max \left\{\sup _{X \in \X} \int_\X|C(X, Y)| \mathrm{d} Y, \sup _{X \in \X} \int_\X |C(Y, X)|\right) \mathrm{d} Y\right\}\\
		&\leq 2 \abs{{\Xpos}} \left(1+\E^{-\beta \tilde{q}}\right) \frac{1}{\tilde{b}} \left(1-\E^{-\Time \tilde{b}} \right).
	\end{aligned}
\end{align}
\end{proof}

Putting this into the context of Theorem \ref{th:main_theorem}, we see that the condition $2 \alpha_C \delta_C^{-2} \eunorm{\V}_{3\delta_C,\X} \leq \frac{1}{2}$ is always fulfilled for small enough times $\Time$. Furthermore, for unitary evolutions, i.e. $B=0$, we see that for small time $\Time$, the connected $n$-point functions can differ at most linearly in $\Time$ compared to a non-interacting system. Moreover, we can quantify what ``small times'' actually means. But in the unitary case $B=0$, our bounds do not allow us to go to arbitrary large times $\Time$. On the other hand, this is possible (not unexpectedly) for the dissipative quantum systems with $B > 0$. 

Generally, taking the thermodynamic limit of large system size, i.e. $\abs{{\Xpos}}\rightarrow \infty$ is not possible. This was to be expected as the general bound given here includes cases of (spatially) non-local interactions. Thus, the above bounds may be seen as a sort of trade-off between generality and optimal bounds for a specific class of systems. 

The results for the decay constant may be improved upon for localised interactions. The following Theorem shows one case in which the decay constant can be bounded uniformly in $\abs{{\Xpos}}$, allowing to take the thermdynamic limit of $\abs{{\Xpos}} \rightarrow \infty$. Assuming $\eunorm{\V}_{3\delta_C,\X}$ not to diverge as $\abs{{\Xpos}} \to \infty$, the thermodynamic limit can then be taken for the truncated expectation values.

\begin{lemma} \label{lem:uniform_decay}
	We assume $Q=A=B>0$ and given any metric $d$ on ${\Xpos}$ (this could be lattice distance or similar), we assume
	\begin{align}
		\abs{Q(x,y)}\leq K_{\nu}(1+d(x,y))^{-\nu}
	\end{align}
	for some $\nu>0$ and $K_{\nu}<\infty$. Let $\Delta>0$ be defined such that $\sigma(Q)\cap[0,\Delta) = \emptyset$ and $n\in \N$, such that $1 \leq n < \nu $. The decay constant $\alpha_C$ and modified decay constant are bounded by
	\begin{align} \label{eq:uniform_bound}
	\begin{aligned}
		\alpha_C &\leq \xi \sqrt{k(2n)}k(\nu -n)^{n} \left( \frac{\Delta}{4}^{-(n+1)} +\frac{\Delta}{4}^{-1}\right) \left( \eunorm{Q} + \Delta \right) \frac{1}{1-\E^{-\beta \frac{\Delta}{2}}} (1+\E^{-\beta \frac{\Delta}{2}})\frac{1}{\Delta}(1-\E^{-\Time \frac{\Delta}{4}}),\\
		\tilde{\alpha}_C &\leq \xi \sqrt{k(2n)}k(\nu -n)^{n} \left( \left(\frac{\Delta}{4}\right)^{-(n+1)} +\left(\frac{\Delta}{4}\right)^{-1}\right) \left( \eunorm{Q} + \Delta \right) \frac{1}{1-\E^{-\beta \frac{\Delta}{2}}}.
		\end{aligned}
	\end{align}
	Here we introduced 
		\begin{align}
			k(\zeta)=\sup _{x\in \Xpos} \sum_{y\in \Xpos} (1+d(x, y))^{-\zeta}, \quad \zeta>0,
		\end{align}
	and a constant $\xi$, which may depend on $n,K_{\nu}$ and $\eunorm{Q}$. \\
	Replacing $\sigma(Q) \cap [0,\Delta) = \emptyset$ by the assumption $\frac{\pi}{2 \beta}\geq \Delta$, similar bounds hold for the case of $A=Q$ and $B=0$.
	\begin{align}
	\begin{aligned}
			\alpha_C &\leq \xi \left( \Delta^{-(n+1)} +\Delta^{-1}\right) \left( \eunorm{Q} + 2 \Delta \right) \frac{1}{\Delta} \left(\E^{\Delta \Time} -1 \right) \sqrt{k(2n)}k(\nu -n)^{n},\\
			\tilde{\alpha}_C &\leq \xi \left( \Delta^{-(n+1)} +\Delta^{-1}\right) \left( \eunorm{Q} + 2 \Delta \right) \frac{1}{\Delta} \E^{\Delta \Time} \sqrt{k(2n)}k(\nu -n)^{n},
	\end{aligned}
	\end{align}
		where again $\xi$ is a constant depending on $n,K_{\nu}$ and $\eunorm{Q}$.
\end{lemma}

\begin{proof}
	The proof boils down to a Combes-Thomas estimate and follows the proof in \cite[Appendix~B]{dRS} for details. Here we treat the case $B=0$. The proof for $B>0$ is very similar. By the spectral theorem, we may write
\begin{align}
	f_{\beta}(\pm Q) \E^{i (t'-t) Q} = \frac{1}{2 \pi i} \oint_{\Gamma} \rd z \E^{i (t'-t) z} f_{\beta}(\pm z) \frac{1}{z-Q},
\end{align}
where $\Gamma$ is some curve encircling the spectrum once counterclockwise. For our purpose we choose a rectangle characterized by the sides $\pm \left(\eunorm{Q} + \Delta \right) \times \left[-\Delta, \Delta \right]$. As $\Delta \leq \frac{\pi}{\beta 2}$, $f_{\beta}(\pm z)\leq 1$ for all $z\in \Gamma$. This also ensures that no pole of $f_{\beta}(\pm z)$ lies inside the rectangle.\\
Analogously to the proof in \cite[Appendix~B]{dRS}, we receive
\begin{align}\label{eq:3}
	\begin{aligned}
	\sup_{x \in \Xpos} \int_{{\Xpos}} \rd y \abs{\bracket{x}{f_{\beta}(\pm Q) \E^{i (t'-t)Q} y}} \leq& \xi \sqrt{k(2n)}k(\nu-n)^{n}\\
	&\times \oint_{\Gamma} \abs{dz} \abs{f_{\beta}(\pm z)}\abs{\E^{i (t'-t) z} \left( \frac{1}{d(z,\sigma(Q))^{n+1}} + \frac{1}{d(z,\sigma(Q))} \right)}. 
\end{aligned}
\end{align}
The contour integral can be bounded by
\begin{align}
\begin{aligned}
	\oint_{\Gamma} \abs{dz} \abs{f_{\beta}(\pm z)}\abs{\E^{i (t'-t) z} \left( \frac{1}{d(z,\sigma(Q))^{n+1}} + \frac{1}{d(z,\sigma(Q))} \right)} &\leq \xi \E^{\abs{t'-t} \Delta} \left( \Delta^{-(n+1)}+ \Delta^{-1} \right)(\eunorm{Q}+2 \Delta).
	\end{aligned}
\end{align}
Taking the supremum and integral with respect to $t,t'$ is equivalent and it explicitly gives
\begin{align}
	\sup_{t \in [0,\Time ]} \int_{0}^{\Time} \rd t' \E^{\abs{t'-t}\Delta} = \frac{1}{\Delta} \left( \E^{\Time \Delta} -1 \right).
\end{align}
Putting everything together and absorbing everything into the constant factor $\xi$
	\begin{align}
		\begin{aligned}
		\eunorm{C}_\infty &\leq \sup_{\sigma \in \{+,-\}} \sum_{\sigma} \sup_{t \in[0,\Time ]} \int \rd t' \sup_{x \in \Xpos} \int_{\Xpos} \rd y \abs{\bracket{x}{f_{\beta}(\pm Q) \E^{i (t'-t) Q}y}}\\
		&\leq \xi \left( \Delta^{-(n+1)} +\Delta^{-1}\right) \left( \eunorm{Q} + 2 \Delta \right) \frac{1}{\Delta} \left(\E^{\Delta \Time} -1 \right) \sqrt{k(2n)}k(\nu -n)^{n}.
		\end{aligned}
	\end{align}
	We note that $k(\zeta)$ is invariant under exchange of the roles of $x$ and $y$, as are all arguments made in the proof of the Theorem. Thus the bound holds for $\eunorm{C}_1$, which concludes the proof.
\end{proof}

\begin{remark}
The conditions in Lemma \ref{lem:uniform_decay} can be loosened to include the case $A=\lambda_A Q$ and $B= \lambda_B Q$ with $Q>0$ and $\lambda_A,\lambda_B>0$. However, the resulting bounds for the decay constant are even lengthier than the ones obtained above.
\end{remark}

\noindent
{\bf Acknowledgement. } This work is supported by Deutsche Forschungsgemeinschaft (DFG, German Research Foundation) under Germany's Excellence Strategy  EXC-2181/1 - 390900948 (the Heidelberg STRUCTURES Cluster of Excellence).

\end{document}